\documentclass[11pt]{article}
\usepackage{etex}
\usepackage{amsmath}
\usepackage{fancyhdr}
\usepackage{amssymb}
\usepackage{amsthm}
\usepackage{graphicx}
\usepackage{varioref}
\usepackage{verbatim} 
\usepackage{multicol}
\usepackage{lmodern}

\usepackage{enumerate}
\usepackage[normalem]{ulem}
\usepackage{caption}
\usepackage{subcaption}
\usepackage[T1]{fontenc}
\usepackage[margin=1in]{geometry}
\usepackage{fancyhdr}
\usepackage{authblk}
\usepackage{paralist}

\usepackage{mathrsfs}

\usepackage{url}
\usepackage{xspace}
\usepackage{thm-restate}

\usepackage{hyperref}
\hypersetup{
    bookmarksnumbered=true, 
    unicode=false, 
    pdfstartview={}, 
    pdftitle={}, 
    pdfauthor={}, 
    pdfsubject={}, 
    pdfcreator={}, 
    pdfproducer={}, 
    pdfkeywords={}, 
    pdfnewwindow=true, 
    colorlinks=true, 
    linkcolor=blue, 
    citecolor=blue, 
    filecolor=blue, 
    urlcolor=blue 
}

\newtheorem{theorem}{Theorem}

\newtheorem{lemma}[theorem]{Lemma}
\newtheorem{proposition}[theorem]{Proposition}
\newtheorem{corollary}[theorem]{Corollary}

\theoremstyle{definition}

\newtheorem{definition}[theorem]{Definition}
\newtheorem{remark}[theorem]{Remark}
\theoremstyle{remark}

\usepackage{tikz} 

\input{Qcircuit}
\makeatother

\newcommand\EXP{{\sf{EXP}}}

\newcommand\DET{{\sf{DET}}}
\newcommand\PLclass{{\sf{PL}}}
\newcommand\QMA{{\sf{QMA}}}

\newcommand\QSPACE{{\sf{Q_{U}SPACE}}}
\newcommand\DSPACE{{\sf{SPACE}}}
\newcommand\QIP{{\sf{QIP}}}

\newcommand\BQSPACE{{\sf{BQSPACE}}}
\newcommand\PSPACE{{\sf{PSPACE}}}
\newcommand\BQPSPACE{{\sf{BQ_{U}PSPACE}}}

\newcommand\PP{\sf{PP}}

\newcommand\NP{{\sf{NP}}}

\newcommand\BQP{{\sf{BQP}}}
\newcommand\PostBQP{{\sf{PostBQP}}}
\newcommand\PQP{{\sf{PQP}}}
\newcommand\Logspace{{\sf{L}}}
\newcommand\NL{{\sf{NL}}}
\newcommand\BPL{{\sf{BPL}}}

\newcommand\preciseQMA{{\sf{PreciseQMA}}}
\newcommand\bddQMA[5]{{\left(#1,#2\right)}\textit{-bounded }\QMA_{#3}(#4,#5)}

\newcommand\preciselh{\textit{Precise Local Hamiltonian}}
\newcommand\preciseklh{\textit{Precise }$k$\textit{-Local Hamiltonian}}
\newcommand\preciseilh[1]{\textit{Precise }\ensuremath{#1}\textit{-Local Hamiltonian}}
\newcommand\matrixinvert[1]{{\ensuremath{#1}}\textit{-Well-conditioned Matrix Inversion}}
\newcommand\spechamiltonian[1]{\ensuremath{#1}\textit{-Minimum Eigenvalue}}
\newcommand\qca[1]{\ensuremath{#1}\textit{-Quantum Circuit Acceptance}}
\newcommand{\zero}{\ensuremath{0^{\otimes{k(n)}}}}

\newcommand{\classfont}{\sf}
\newcommand{\Unitary}{\mathbf{U}}
\newcommand{\unitaryBQL}{{\classfont{BQ}_\Unitary\classfont{L}}}

\newcommand{\unitaryQSPACE}[3]{{\classfont{Q}_\Unitary\classfont{SPACE}}[#1](#2,#3)}
\newcommand{\unitaryBQSPACE}[1]{{\classfont{BQ}_\Unitary\classfont{SPACE}}[#1]}
\newcommand\bigoh{\mathcal{O}}

\newcommand{\polylog}{\textrm{polylog}}

\usepackage{xcolor}

\DeclareMathAlphabet{\matheu}{U}{eus}{m}{n}

\DeclareMathOperator{\tr}{tr}

\DeclareMathOperator{\poly}{poly}

\usepackage{mathtools}

\begin{document}

\begin{titlepage}
\title{A Complete Characterization of Unitary Quantum Space}
\author[1]{Bill Fefferman\thanks{wjf@umd.edu}}
\author[1]{Cedric Yen-Yu Lin\thanks{cedricl@umiacs.umd.edu}}
\affil[1]{Joint Center for Quantum Information and Computer Science (QuICS), University of Maryland}
\date{\today}
\maketitle
\begin{abstract}

Motivated by understanding the power of quantum computation with restricted number of qubits, we give two complete characterizations of unitary quantum space bounded computation.  First we show that approximating an element of the inverse of a well-conditioned efficiently encoded $2^{k(n)}\times 2^{k(n)}$ matrix is complete for the class of problems solvable by quantum circuits acting on $\bigoh(k(n))$ qubits with all measurements at the end of the computation.  Similarly, estimating the minimum eigenvalue of an efficiently encoded Hermitian $2^{k(n)}\times 2^{k(n)}$ matrix is also complete for this class. In the logspace case, our results improve on previous results of Ta-Shma \cite{tashma} by giving new space-efficient quantum algorithms that avoid intermediate measurements, as well as showing matching hardness results.

Additionally, as a consequence we show that $\preciseQMA$, the version of $\QMA$ with exponentially small completeness-soundess gap, is equal to $\PSPACE$.  Thus, the problem of estimating the minimum eigenvalue of a \emph{local} Hamiltonian to inverse exponential precision is $\PSPACE$-complete, which we show holds even in the frustration-free case.  Finally, we can use this characterization to give a provable setting in which the ability to prepare the ground state of a local Hamiltonian is more powerful than the ability to prepare PEPS states.

Interestingly, by suitably changing the parameterization of either of these problems we can completely characterize the power of quantum computation with \emph{simultaneously} bounded time and space.

\end{abstract}
\thispagestyle{empty}
\end{titlepage}
\clearpage
\setcounter{page}{1}
\section{Introduction}
How powerful is quantum computation with a restricted number of qubits? In this work we will study unitary quantum space-bounded classes - those problems solvable using a given amount of (quantum and classical) space, with all quantum measurements performed at the end of the computation. We give two sets of complete problems for these classes; to the best of our knowledge, these are the first natural complete problems proposed for quantum space-bounded classes. 

The first problem we consider, the $\matrixinvert{k(n)}$ problem, is a well-conditioned version of the ubiquitous matrix inversion problem. The second problem we consider, the $\spechamiltonian{k(n)}$ problem, asks us to compute the minimum eigenvalue of a Hermitian matrix to high precision -- in the context of quantum complexity, this is a natural generalization of the familiar local Hamiltonian problem \cite{ksv02}. Interestingly enough, the first (resp. second) problem is the space-bounded variant of a $\BQP$-complete \cite{HHL} (resp. $\QMA$-complete \cite{ksv02}) problem; their complexities coincide in the space-bounded regime. For the sake of readability, we defer precise definitions of these problems and statements of our results until Sections \ref{sec: Matrix Inversion} and \ref{sec: Minimum Eigenvalue}.

We now proceed to give some justification for the importance of our results.  In the following discussion, $\unitaryBQSPACE{k(n)}$ refers to the class of problems solvable with bounded error by a quantum algorithm running in $\mathcal{O}(k(n))$; the subscript $\Unitary$ indicates that the algorithm is unitary, i.e. employs no intermediate measurements.

\subsection{Background and Motivation}
The Matrix Inversion problem is of central importance in computational complexity theory.  Matrix inversion is known to be complete for $\DET$, the class of functions as hard as computing the determinant of an integer matrix, which can be solved in classical $\mathcal{O}(\log^{2}(n))$ space \cite{berkowitz, cook}. It is a major open problem to determine if Matrix Inversion can be solved in classical logarithmic space, which would imply $\Logspace=\NL=\DET$.
 
Recently, Ta-Shma \cite{tashma}, building on work of Harrow, Hassidim, and Lloyd \cite{HHL}, showed that a well-conditioned $n \times n$ matrix can be inverted (up to $1/\poly(n)$ error) by a quantum $O(\log n)$ space algorithm using intermediate measurements. Similarly, Ta-Shma also gives an algorithm for computing eigenvalues of a Hermitian matrix with similar space. These algorithms achieve a quadratic advantage in space over the best known classical algorithms, which require $\Omega(\log^2n)$ space.  This is the maximum quantum advantage possible, since Watrous has shown $\BQSPACE[k(n)]\subseteq\DSPACE[\bigoh(k(n)^2)]$ \cite{Watrous99,Watrous03} even for quantum algorithms with intermediate measurements.

Our completeness result for matrix inversion, along with observing our algorithm for matrix inversion (Theorem \ref{thm: matrix inversion alg}) actually gives a high-precision approximation, gives the following corollary in the logspace case (see Remark \ref{rem:logspace}).
\begin{corollary}
The problem of approximating, to constant precision, an entry of the inverse of an $n \times n$ positive semidefinite matrix with condition number at most $\poly(n)$ is $\unitaryBQL$-complete under $\Logspace$-reductions, where $\unitaryBQL$ is the set of problems solvable in unitary quantum logspace. This problem remains in $\unitaryBQL$ even if $1/\poly(n)$ precision is required.
\end{corollary}

Similarly, restricting Thereom \ref{sec: Minimum Eigenvalue} to the logspace case gives the following corollary.
\begin{corollary}
The problem of approximating, to $1/\poly(n)$ precision, the minimum eigenvalue of an $n \times n$ positive semidefinite matrix is $\unitaryBQL$-complete under $\Logspace$-reductions.
\end{corollary}

These corollaries improve upon Ta-Shma's results \cite{tashma} in two ways. First, our algorithms  solve these problems without needing intermediate measurements.  Unlike in time complexity, where the ``Principle of safe storage'' gives a time-efficient procedure to defer intermediate measurements, these methods may incur an exponential blow-up in space.  

One might wonder why we care so much about avoiding intermediate measurements.  The main reason is that removing intermediate measurements from the computation allows us to give matching hardness results, showing the optimality of our algorithms.  This is the second way our results improve on those of Ta-Shma.  In particular, our proofs crucially use space-efficient methods for the amplification of unitary quantum computations, which are not known in the non-unitary model.  This is because the techniques require applying the inverse of the circuit, which of course is impossible if the circuit contains intermediate measurements.  We will also rely on ideas from Kitaev's clock construction, which constructs a local Hamiltonian from a unitary quantum circuit. 

Specifically, we will show that the problems of inverting well-conditioned matrices and computing minimum eigenvalues of Hermitian matrices are hard for unitary quantum logspace under $\Logspace$-reductions. In the case of our algorithm for Matrix Inversion, this means that the upper bound on the condition number bound is unlikely to be improved upon.  Likewise, this gives some of the strongest evidence that even well-conditioned matrices cannot be inverted in deterministic logspace, since otherwise our results would immediately imply $\Logspace=\unitaryBQL$, which seems quite unlikely. 

Interestingly, although our algorithms for both problems use different techniques from those of Ta-Shma, our algorithm for computing the minimum eigenvalue is completely different.  In particular, our algorithm crucially relies on new methods for space efficient $\QMA$ amplification, together with some of the most powerful recent results in Hamiltonian simulation \cite{berry14,berry15}.

Concurrently with our work, Doron, Sarid, and Ta-Shma have shown that analogous problems for stochastic matrices (e.g. computing the eigenvalue gap) are complete for classical randomized logspace, or $\BPL$ \cite{dt15,dst16}. In addition, Le Gall has shown that analogous problems for Laplacian matrices can be solved in $\BPL$ \cite{legall16}. Since it is straightforward to see that Well-conditioned Matrix Inversion reduces to Integer Matrix Inversion, we obtain a direct proof that $\unitaryBQL\subseteq\DET$, which was previously known indirectly via the containments $\unitaryBQL\subseteq\PLclass\subseteq\DET$ \cite{Watrous03,Borodin84}.

Therefore the power of classical and quantum space-bounded classes are characterized by the ability to approximate solutions of different problems in $\DET$ (stochastic matrices for the former, and Hermitian matrices for the latter). This could shed light on the differences between deterministic, randomized, and quantum space complexity. An open question is to find a class of interesting matrices whose inverse (or eigenvalues) can be computed in \emph{deterministic} logspace.

Interestingly, if we change the scaling of the parameters in our Matrix Inversion and Minimum Eigenvalue problems suitably, then we obtain problems that are known to be complete for $\BQP$ \cite{HHL} and $\QMA$ \cite{ksv02,at03}.
Thus by appropriately bounding the dimension of the matrix and either the condition number or the promise gap, we can give problems complete for quantum time or quantum space. In fact we can strengthen these results to settings with a simultaneously bounded amount of space and time; see Section \ref{sec: time and space}. 
\subsection{Relationship with Matchgates} \label{sec: matchgates}

Matchgates are a subclass of quantum gates introduced by Valiant \cite{Val02SIComp}, who also showed that nearest neighbor matchgate circuits (which we will just call matchgate computations) are classically simulable. Matchgate computations were further shown to be equivalent to a one-dimensional model of noninteracting fermions by Terhal and DiVincenzo \cite{TerDiV02PRA}; and equivalent to unitary quantum logspace by Jozsa, Kraus, Miyake, and Watrous \cite{JozKraMiyWat10RSPA}. Our complete problems therefore elucidate the computational power of noninteracting fermions.

We know that sampling from output distributions of matchgate computations gives us the power of $\BPL$; but what is the computational power of computing exactly the output probabilities of matchgate computations? We conjecture that this computational power corresponds to $\DET$, since amplitudes of noninteracting fermion circuits are related to determinants (and see also the discussion in the previous subsection). It is known that output probabilities of matchgate computations can be exactly calculated by an efficient classical algorithm \cite{JozMiy08RSPA}, which is consistent with our conjecture because $\DET \in \classfont{P}$.

\subsection{Quantum Merlin-Arthur with Small Gap}
A consequence of our proof of completeness for the $\spechamiltonian{k(n)}$ problem is an equivalence between space-bounded quantum computations and quantum Merlin-Arthur proof systems. Here we give this equivalence for the polynomial space case: let $\preciseQMA$ be the variant of $\QMA$ with exponentially small completeness-soundness gap. Then we show the following:
\begin{corollary} \label{cor: preciseqma}
$\preciseQMA = \BQPSPACE =  \PSPACE$.
\end{corollary}
The second equality is due to Watrous \cite{Watrous99,Watrous03}. We give similar equivalences for space-bounded quantum computations with and without a witness for other space bounds as well (Theorem \ref{thm: equivalence}).

We note that $\preciseQMA$ is likely far more powerful than its classical counterpart. The analogous classical complexity class is contained in $\NP^{\PP}$: given a classical witness, the verifier runs a classical computation that in the YES case accepts with probability at least $c$, or in the NO case accepts with probability at most $s$, where $c>s$. Note that in the classical case $c - s > \exp(-\poly)$ is automatically satisfied. Since $\NP^{\PP}$ is in the counting hierarchy, the entirety of which is contained in $\PSPACE$ (see e.g., \cite{allenderwagner}), we see that the quantum proof protocol is strictly stronger than the classical one, unless the counting hierarchy collapses to the second level.

We also show that the \emph{local} Hamiltonian problem is $\PSPACE$-complete when the promise gap is exponentially small (for details see Appendix \ref{app:localhamiltonian}). This is in contrast to the usual case when the gap is polynomially small, where the problem is $\QMA$-complete. Perhaps more surprisingly, $\preciseQMA = \PSPACE$ is more powerful than $\PostBQP=\PP$, the class of problems solvable with postselected quantum computation \cite{aaronson05}.

Another consequence concerns Projected Entangled Pair States, or PEPS, a natural extension of matrix product states to two and higher spatial dimensions, which can be described as the ground state of certain frustration-free local Hamiltonians \cite{vc04}. A characterization of the computational power of PEPS was given in \cite{swv07}, and can be summarized as follows: let $O_{PEPS}$ be a quantum oracle that, given the description of a PEPS, outputs the PEPS (so the output is quantum). Then $\BQP^{O_{PEPS}}_{\parallel,\text{classical}} = \PostBQP = \PP$, where (following Aaronson \cite{aaronson05}) the subscript  denotes that only classical nonadaptive queries to the oracle are allowed. Moreover, let $\PQP$ stand for the set of problems solvable by a quantum computer with \emph{unbounded error}; then it can be straightforwardly shown that $\PQP^{O_{PEPS}}_{\parallel,\text{classical}} = \PP$ as well (see Appendix \ref{app:peps}).

On the other hand, suppose we have an oracle $O_{LH}$ that given the description of a local Hamiltonian, outputs a ground state of the Hamiltonian. Then our results show that $\preciseQMA = \PSPACE \subseteq \PQP^{O_{LH}}_{\parallel,\text{classical}}$. This shows that in the setting of unbounded-error quantum computation, PEPS do not capture the full computational complexity of general local Hamiltonian ground states unless $\PP=\PSPACE$. We leave open the problem of determining the complexity of $\BQP^{O_{LH}}_{\parallel,\text{classical}}$.

Lastly, we are able to strengthen our characterization to show that $\preciseQMA$ contains $\PSPACE$ (see Appendix \ref{app:perfectcompleteness}), even when restricted to having perfect completeness.  This allows us to prove that testing if a local Hamiltonian is frustration-free is a $\PSPACE$-complete problem (Appendix \ref{app:localhamiltonian}). We note that if the local Hamiltonian is promised to have a ground state energy of at least $1/\poly$ if it is frustrated, then this is the Quantum Satisfiability problem defined by Bravyi, which is known to be $\QMA_1$ complete \cite{bravyi06,gn13}. Our results show that if the promise gap is removed then we instead get $\PSPACE$-completeness.

 \section{Preliminaries}
\subsection{Quantum circuits}
We will assume a working knowledge of quantum information. For an introduction, see \cite{nc00}.

A \emph{quantum circuit} consists of a series of quantum gates each taken from some universal gateset, such as the gateset consisting of Hadamard and Toffoli gates \cite{shi}. 
 For functions $f,g:\mathbb{N}\rightarrow\mathbb{N}$, we say a family of quantum circuits $\{Q_x\}_{x\in\{0,1\}^*}$ is \emph{$f$-time $g$-space uniformly generated} if there exists a deterministic classical Turing machine that on input $x\in\{0,1\}^n$ and $i>0$ outputs the $i$-th gate of $Q_x$ within time $f(n)$ and workspace $g(n)$ \cite{nc00}.  
 
Our restriction to a specific gateset is without loss of generality, even for logarithmic space algorithms:  there exists a deterministic algorithm that given any unitary quantum gate $U$ and a parameter $\epsilon$, outputs a sequence of at most $\polylog({1/\epsilon})$ gates from any universal quantum gateset that approximates $U$ to precision $\epsilon$ in space $\bigoh({\log({1/\epsilon})})$ and time $\polylog({1/\epsilon})$ \cite{mw12}.  This improves the Solovay-Kitaev theorem, which guarantees a space bound of $\polylog(1/\epsilon)$; see e.g., \cite{nc00}.
\subsection{Space-bounded computation} \label{sec: space bounded computation}
For our model of unitary quantum space-bounded computation, we consider a quantum system with purely classical control, because there are no intermediate quantum measurements to condition future operations on. Specifically, we use the following definition (see Appendix \ref{app: space bounded} for more details):

\begin{definition} \label{def: qspace}
Let $k(n)$ be a function satisfying $\Omega(\log(n)) \le k(n) \le \poly(n)$. A promise problem $L=(L_{yes},L_{no})$ is in $\QSPACE[k(n)](c,s)$ if there exists a $\poly(|x|)$-time $\mathcal{O}(k)$-space uniformly generated family of quantum circuits $\{Q_x\}_{x\in\{0,1\}^*}$, where each circuit $Q_x=U_{x,T}U_{x,T-1}\cdots U_{x,1}$ has $T=2^{\mathcal{O}(k)}$ gates, and acts on $\mathcal{O}(k(|x|))$ qubits, such that:\\
 If $x \in L_{yes}$:
\begin{equation}
\bra{0^k} Q^\dagger_x \ket{1}\bra{1}_{out} Q_x \ket{0^k} \ge c.
\end{equation}
Whereas if $x \in L_{no}$:
\begin{equation}
\bra{0^k} Q^\dagger_x \ket{1}\bra{1}_{out} Q_x \ket{0^k} \le s.
\end{equation}
Here $out$ denotes a single qubit we measure at the end of the computation; no intermediate measurements are allowed.  
Furthermore, we require $c$ and $s$ to be computable in classical $\bigoh(k(n))$-space.
\end{definition}

For the rest of the paper we will always assume that $\Omega(\log(n)) \le k(n) \le \poly(n)$.

The bound $T=2^{\mathcal{O}(k)}$ on the circuit size comes from that any classical Turing machine generating $Q_x$ using space $\bigoh(k(|x|))$ has at most $2^{\mathcal{O}(k)}$ configurations. We note that $2^{\mathcal{O}(k)}$ gates suffice to approximate any gate on $\bigoh(k)$ qubits to high accuracy (see e.g. \cite[Chapter~4]{nc00}). The $\poly(|x|)$ time bound on the classical control can be assumed without loss of generality; see Appendix \ref{app: space bounded}.

\begin{definition} $\unitaryBQSPACE{k}=\unitaryQSPACE{k}{2/3}{1/3}$.\end{definition}
\begin{theorem}[Watrous \cite{Watrous99,Watrous03}]\label{thm:pqpspace} $\unitaryBQSPACE{\poly}=\PSPACE$.
\end{theorem}

We now define space- and time-bounded analogues of $\QMA$:
\begin{definition}We say a promise problem $L=(L_{yes},L_{no})$ is in $\bddQMA{t}{k}{m}{c}{s}$ if there exists a $t$-time and $(k+m)$-space uniformly generated family of quantum circuits $\{ V_x\}_{x\in\{0,1\}^*}$, each of size at most $t(|x|)$, acting on $k(|x|)+m(|x|)$ qubits, so that:\\

If $x \in L_{yes}$ there exists an $m$-qubit state $\ket{\psi}$ such that:
\begin{equation}
\left(\bra{\psi}\otimes \bra{0^k}\right) V^\dagger_x \ket{1}\bra{1}_{out} V_x \left(\ket{\psi}\otimes \ket{0^k}\right) \ge c.
\end{equation}
Whereas if $x \in L_{no}$, for all $m$-qubit states $\ket{\psi}$ we have:
\begin{equation}
\left(\bra{\psi}\otimes \bra{0^k}\right) V^\dagger_x \ket{1}\bra{1}_{out} V_x \left(\ket{\psi}\otimes \ket{0^k}\right) \le s.
\end{equation}
$out$ denotes a single qubit measured at the end of the computation; no intermediate measurements are allowed.  
Here $c$ and $s$ are computable in classical $\bigoh(t(n))$-time and $\bigoh(k(n)+m(n))$-space.
  \end{definition}

\begin{definition} $\QMA=\bddQMA{\poly}{\poly}{\poly}{2/3}{1/3}$.
\end{definition}
\begin{definition} $\preciseQMA=\bigcup_{c \in (0,1]}\bddQMA{\poly}{\poly}{\poly}{c}{c-2^{-\poly}}$.
\end{definition}

\subsection{Other definitions and results}
We use the following definition of \emph{efficient encodings} of matrices:
\begin{definition}\label{def: efficient encoding}Let $M$ be a $2^{k} \times 2^{k}$ matrix, and $\mathcal{A}$ be a classical algorithm (e.g. a Turing machine) specified using $n$ bits. We say that $\mathcal{A}$ is an \emph{efficient encoding} of $M$ if on input $i\in\{0,1\}^k$, $\mathcal{A}$ outputs the indices and contents of the non-zero entries of the $i$-th row, using at most $\poly(n)$ time and $\bigoh(k)$ workspace (not including the output size). Note that as a consequence $M$ has at most $\poly(n)$ nonzero entries in each row.
\end{definition}
We will often specify a matrix $M$ in the input by giving an efficient encoding of $M$. The size of the encoding is then the input size, which we will usually indicate by $n$. 
\begin{remark} \label{rem:logspace}
It is not difficult to see that any $n \times n$ matrix has an efficient encoding of size $\bigoh(n)$, since it is straightforward to construct a classical $\bigoh(\log n)$-size circuit that on input $i,j$ outputs the $(i,j)$-entry of the matrix.
\end{remark}

In our results we will implicitly assume the existence of algorithms that compute some common functions on $n$-bit numbers, such as $\sin,\cos,\arcsin,\arccos$ and exponentiation, to within $1/\poly(n)$ accuracy in classical $\bigoh (\log{n})$ space.  Algorithms for these tasks have been designed by Reif \cite{reif}.\footnote{Reif's algorithms take only $\bigoh (\log\log{n}\log\log\log{n})$ space, but for simplicity we only use the $\bigoh (\log{n})$ bound.}

Finally, we will need new results from the Hamiltonian simulation literature:
\begin{theorem}[\cite{berry14,bccks15,berry15}] \label{thm:ham_sim}
Given as input is the size-$n$ efficient encoding of a $2^{k(n)} \times 2^{k(n)}$ Hermitian matrix $H$. Then treated as a Hamiltonian, the time evolution $\exp(-iHt)$ can be simulated using $\text{\emph{poly}}(n,k,\|H\|_{max}, t, \log(1/\epsilon))$ operations and $\bigoh(k+\log(t/\epsilon))$ space.
\end{theorem}
While the space complexity was not explicitly stated in \cite{berry14,bccks15,berry15}, it can be seen from the analysis (see e.g. \cite{bccks15}). The crucial thing to notice in Theorem \ref{thm:ham_sim} is the polylogarithmic scaling in the error $\epsilon$; this implies that we can obtain polynomial precision in $\exp(-iHt)$ using only polynomially many operations. Also note that the maximum eigenvalue of $H$, $\|H\|$, satisfies $\|H\| \le \poly(n) \|H\|_{max}$.

\section{The Well-Conditioned Matrix Inversion Problem} \label{sec: Matrix Inversion}

We begin with a formal statement of the problem:
\begin{definition}[\matrixinvert{k(n)}] \label{def: matrix invert}
Given as input is the size-$n$ efficient encoding of a $2^{k(n)} \times 2^{k(n)}$ positive semidefinite matrix $H$ with a known upper bound $\kappa = 2^{\mathcal{O}(k(n))}$ on the condition number, so that $\kappa^{-1}I\preceq H \preceq I$, and $s,t\in \lbrace 0,1\rbrace^{k(n)}$. It is promised that either $|H^{-1}(s,t)|\geq b$
 or $|H^{-1}(s,t)|\leq a$ for some constants $0 \le a < b \le 1$; determine which is the case.
 \end{definition}

\begin{theorem} \label{thm: matrix invert}
For $\Omega(\log(n)) \le k(n) \le \poly(n)$, $\matrixinvert{\mathcal{O}(k(n))}$ is complete for $\unitaryBQSPACE{\mathcal{O}(k(n))}$ under classical reductions using $\poly(n)$ time and $\mathcal{O}(k(n))$ space.
\end{theorem}
\begin{proof}
We begin by giving a new space efficient algorithm for this matrix inversion problem:

\begin{theorem} \label{thm: matrix inversion alg}
Fix functions $k(n)$, $\kappa(n)$, and $\epsilon(n)$. Suppose we are given the size-$n$ efficient encoding of a $2^{k(n)} \times 2^{k(n)}$ PSD matrix $H$ such that $\kappa^{-1} I \preceq H \preceq I$. We are also given $\poly(n)$-time $\mathcal{O}(k+\log(\kappa/\epsilon))$-space uniform quantum circuits $U_a$ and $U_b$ acting on $k$ qubits and using at most $T$ gates. Let $U_a\ket{0}^{\otimes k(n)} = \ket{a}$ and $U_b\ket{0}^{\otimes k(n)} = \ket{b}$. The following tasks can be performed with $\poly(n)$-time $O(k+\log(\kappa/\epsilon))$-space uniformly generated quantum circuits with $\poly(T,k,\kappa,1/\epsilon)$ gates and $\bigoh (k+\log(\kappa/\epsilon))$ qubits:
\begin{compactenum}
\item With at least constant probability, output an approximation of the quantum state $H^{-1}\ket{b} / \|H^{-1}\ket{b}\|$ up to error $\epsilon$.
\item Approximate $\|H^{-1}\ket{b}\|$ to precision $\epsilon$.
\item Approximate $|\bra{a}H^{-1}\ket{b}|$ to precision $\epsilon$.
\end{compactenum}
These circuits do not require intermediate measurements.
\end{theorem}

In fact our algorithm is much stronger: to solve $\matrixinvert{k(n)}$ we merely need to approximate $|\bra{s}H^{-1}\ket{t}|$ to constant precision, while Theorem \ref{thm: matrix inversion alg} actually gives an approximation to precision $2^{-\bigoh (k)}$ in $\bigoh (k(n))$ unitary quantum space. Moreover our algorithm does not require $s$ and $t$ to be computational basis states.

We note that we can modify our definition of unitary quantum space-bounded classes to include computing functions, for instance by adding a write-only one-way output tape of qubits to the Turing machine (see the discussion in Appendix \ref{app: space bounded}), that are all measured at the very end of the computation. The error reduction result (Corollary \ref{obvious2}) later in our work allows the total error to be reasonably controlled. With such a modification we can compute the whole matrix inverse in unitary quantum logspace. We will not pursue this modified model further in this work.

\begin{proof}
We first briefly summarize the algorithm of Ta-Shma \cite{tashma}, which  is based on the linear systems solver of Harrow, Hassidim and Lloyd \cite{HHL}.  Ta-Shma shows that an $n \times n$ matrix with condition number at most $\poly(n)$ can be inverted by a quantum logspace algorithm with intermediate measurements; in our language this corresponds to solving $\matrixinvert{\bigoh(\log{n})}$. 

Our algorithm and Ta-Shma's share the same initial procedure.  In particular it is shown:  
\begin{lemma}[Implicit in \cite{HHL,tashma}] \label{lem: matrix inversion lemma}
There is a $\poly$-time $\bigoh (k + \log(\kappa/\epsilon'))$-space uniform quantum unitary transformation $W_H$ over $k+\ell = \bigoh (k+ \log(\kappa/\epsilon'))$ qubits and using $\poly(k,\kappa/\epsilon')$ gates, such that for any $k$-qubit input state $\ket{b}$,
\begin{equation}
W_H (\ket{0}^{\otimes \ell} \otimes \ket{b}) = \alpha \ket{0}_{out} \otimes \ket{\psi_b} + \sqrt{1-\alpha^2} \ket{1}_{out} \otimes\ket{\psi'_b},
\end{equation}
where $\ket{\psi_b}$ and $\ket{\psi'_b}$ are normalized states such that $\| \ket{\psi_b} - \ket{0}^{\otimes \ell-1} \otimes \frac{H^{-1}\ket{b}} {\|H^{-1}\ket{b}\|} \| \le \epsilon'$, $\alpha$ is a positive number satisfying $|\alpha - \frac{\|H^{-1}\ket{b}\|}{\kappa}| \le \epsilon'$, and ``out'' is a 1-qubit register.
\end{lemma}
This lemma can be obtained by combining the Hamiltonian simulation algorithms of Berry et al. (Theorem \ref{thm:ham_sim}) with the analysis of Harrow, Hassidim and Lloyd \cite{HHL}; a version without the time bound is implicit in the proof of \cite[Theorem~6.3]{tashma}. For completeness, we sketch the proof below.
\begin{proof}[Proof sketch]
Decompose $\ket{b}$ into the eigenbasis of $H$: $\ket{b} = \sum_{\lambda} a_\lambda \ket{v_\lambda}$, where $\lambda$ are eigenvalues of $H$ and $H\ket{v_\lambda} = \lambda \ket{v_\lambda}$. The following procedure satisfies Lemma \ref{lem: matrix inversion lemma} (all steps are approximate):
\begin{enumerate}
\item Perform phase estimation on the operator $\exp(iH)$ and state $\ket{b}$ to compute the eigenvalues of $H$ into an ancilliary register, obtaining the state $\sum_{\lambda} a_\lambda \ket{v_\lambda}\ket{\lambda}$.
\item Implement the unitary transformation $\ket{\lambda}\ket{0} \rightarrow \ket{\lambda}[(\kappa\lambda)^{-1}\ket{0}+(\sqrt{1-(\kappa\lambda)^{-2}}\ket{1}]$, to obtain the state $\sum_{\lambda} a_\lambda \ket{v_\lambda}\ket{\lambda}[(\kappa\lambda)^{-1}\ket{0}+(\sqrt{1-(\kappa\lambda)^{-2}}\ket{1}]$.
\item Uncompute the eigenvalues $\lambda$ by running phase estimation in reverse, obtaining the state
$\sum_{\lambda} a_\lambda \ket{v_\lambda}\ket{0}^{\ell-1}[(\kappa\lambda)^{-1}\ket{0}+(\sqrt{1-(\kappa\lambda)^{-2}}\ket{1}]$. Note that $\sum_{\lambda} a_\lambda \ket{v_\lambda}\ket{0}^{\ell-1}(\kappa\lambda)^{-1}\ket{0} = \frac{1}{\kappa} H^{-1}\ket{b}$.
\end{enumerate}
An appropriate error analysis of this procedure is the technical bulk of the proof; we refer the reader to \cite{HHL}. For Step 1, Ta-Shma showed how to implement $\exp(iH)$ in $\bigoh (k+\log(1/\epsilon))$ space \cite[Theorem~4.1]{tashma} (their proof works for general matrices with efficient encodings); recent sparse Hamiltonian simulation algorithms (Theorem \ref{thm:ham_sim}) give a time efficient way to do this.
\end{proof}
Intuitively, Lemma \ref{lem: matrix inversion lemma} gives a space-efficient quantum algorithm that produces a state proportional to $H^{-1}\ket{b}$ with probability at least $1/\kappa$.  Our goal is to produce amplify the probability from $1/\kappa$ to a constant, to produce a state with constant overlap to the state $\ket{0}_{out} \otimes \ket{\psi_b}$ together with an estimate for $\alpha \approx \|H^{-1}\ket{b}\|$.  From here our algorithm differs from Ta-Shma's and uses a combination of amplitude amplification and phase estimation.  This sidesteps both the somewhat involved analysis and intermediate measurements of Ta-Shma's algorithm. 

Specifically, consider the two projectors
\begin{equation}
\Pi_0 = \ket{0}\bra{0}^{\otimes \ell} \otimes \ket{b}\bra{b}, \quad \Pi_1 = W_H^\dagger (\ket{0}\bra{0}_{out} \otimes I) W_H.
\end{equation}
$\Pi_0$ projects onto the initial subspace, while $\Pi_1$ projects onto the initial states that would be accepted by the final measurement. The rotation $R=-(I-2\Pi_1)(I-2\Pi_0)$ has eigenvalues $e^{\pm i2\sin^{-1}\alpha}$ with eigenvectors $\ket{\psi_+}$, such that $\ket{0}^{\otimes \ell} \otimes \ket{b} = (\ket{\psi_+}+\ket{\psi_-})/\sqrt{2}$ is a uniform superposition of the two eigenvectors. Therefore phase estimation on the operator $R$ and input state $\ket{0}^{\otimes \ell} \otimes \ket{b}$ suffices to give an estimate of $\alpha$. Furthermore both eigenvectors have constant overlap with $W_H^{\dagger}\ket{\psi_b}$, so applying $W_H$ to the residual state of phase estimation allows us to complete the first task as well.

We have addressed the first two tasks in Theorem \ref{thm: matrix inversion alg}. For the third task (approximating $|\bra{a}H^{-1}\ket{b}|$), we can choose $\Pi'_1 = W_H^\dagger (\ket{0}\bra{0}_{out} \otimes I) (I_{anc} \otimes \ket{a}\bra{a})(\ket{0}\bra{0}_{out} \otimes I) W_H$ instead, and phase estimation on $R=-(I-2\Pi_1)(I-2\Pi'_0)$ will give an estimate for $|\bra{a}H^{-1}\ket{b}|$. See Appendix \ref{app: matrix inversion alg} for the full proof.
\end{proof}
We establish that $\matrixinvert{k(n)}$ is $\unitaryBQSPACE{\mathcal{O}(k)}$-hard using a similar argument to Harrow, Hassidim, and Lloyd \cite{HHL}, in which given a quantum circuit acting on $k(n)$ qubits we 
construct a efficiently encoded well-conditioned $2^{\mathcal{O}(k)}\times 2^{\mathcal{O}(k)}$ matrix $H$, so that a single element of $H^{-1}$ is proportional to the success probability of the circuit.  See Appendix \ref{app: matrixinversion-hardness}.
\end{proof}

\section{The Minimum Eigenvalue Problem} \label{sec: Minimum Eigenvalue}
Our second characterization of unitary quantum space is based on the following problem:
\begin{definition}[$\spechamiltonian{k(n)}$ problem] \label{def: spechamiltonian}
Given as input is the size-$n$ efficient encoding of a $2^{k(n)} \times 2^{k(n)}$ PSD matrix $H$, such that 
$\|H\|_{max} = \max_{s,t}|H(s,t)|$ is at most a constant. Let $\lambda_{min}$ be the minimum eigenvalue of $H$. It is promised that either $\lambda_{min} \le a$ or $\lambda_{min} \ge b$, where $a(n)$ and $b(n)$ are numbers such that $b-a > 2^{-\mathcal{O}(k(n))}$. Output 1 if $\lambda_{min} \le a$, and output 0 otherwise.
\end{definition}

\begin{theorem} \label{thm: spechamiltonian}
For $\Omega(\log(n)) \le k(n) \le \poly(n)$, $\spechamiltonian{\mathcal{O}(k(n))}$ is complete for \\ $\unitaryBQSPACE{\mathcal{O}(k(n))}$ under classical reductions using $\poly(n)$ time and $\mathcal{O}(k(n))$ space.
\end{theorem}

In the process of proving this result, we will also show the following equivalence:
\begin{theorem} \label{thm: equivalence}
$\unitaryBQSPACE{\mathcal{O}(k(n))}$ is equivalent to the class of problems characterized by having quantum Merlin Arthur proof systems running in polynomial time, $\mathcal{O}(k(n))$ witness size and space, and $2^{-\mathcal{O}(k(n))}$ completeness-soundness gap. Or in other words,
\begin{equation}
\unitaryBQSPACE{\mathcal{O}(k(n))}= \bigcup_{c - s \ge 2^{-\bigoh(k(n))}}\bddQMA{\poly}{\bigoh(k(n))}{\bigoh(k(n))}{c}{s} \nonumber
\end{equation}
\end{theorem}

Our proof will consist of three steps. Lemma \ref{lem:qma protocol} will show that $\spechamiltonian{k(n)}$ is in the generalized $\preciseQMA$ class defined in Theorem \ref{thm: equivalence}. Lemma \ref{lem: pspace upper bound} will show that this generalized $\preciseQMA$ class is contained in $\unitaryBQSPACE{k(n)}$. Finally, Lemma \ref{lem: specham-hardness} will show that $\unitaryBQSPACE{k(n)}$-hardness of $\spechamiltonian{k(n)}$.

\begin{lemma} \label{lem:qma protocol}
$\spechamiltonian{k(n)}$ is contained in $\bddQMA{\poly}{\bigoh(k(n))}{\bigoh(k(n))}{c}{s}$ for some $c,s$ such that $c - s > 2^{-\mathcal{O}(k(n))}$.
\end{lemma}
\begin{proof}
We are given the size-$n$ efficient encoding of a $2^{k(n)} \times 2^{k(n)}$ PSD matrix $H$, and it is promised that the smallest eigenvalue $\lambda_{min}$ of $H$ is either at most $a$ or at least $b$. Merlin would like to convince us that $\lambda_{min} \le a$; he will send us a purported $k$-qubit eigenstate $\ket{\psi}$ of $H$ with eigenvalue $\lambda_{min}$. Choose $t = \pi / (\poly(n)\|H\|_{max}) \le \pi / \|H\|$; then all eigenvalues of $Ht$ lie in the range $[0,\pi]$, and the output of phase estimation on $\exp(-iHt)$ will be unambiguous. We perform, on $\psi$, phase estimation of $\exp(-iHt)$ with one bit of precision:
\begin{align} \label{circ: poor man phase estimation}
&\Qcircuit @C=1em @R=.7em {
\lstick{\ket{0}}& \gate{H} & \ctrl{1} & \gate{H} & \rstick{\frac{1+e^{-i\lambda t}}{2}\ket{0} +\frac{1-e^{-i\lambda t}}{2}\ket{1} } \qw \\
\lstick{\ket{\psi}}& \qw & \gate{e^{-iHt}}  & \qw & \rstick{\ket{\psi}} \qw
}
\end{align}
Here the $H$ gates on the first qubit are Hadamard gates (and have nothing to do with the matrix $H$). Theorem \ref{thm:ham_sim} gives an implementation of $\exp(-iHt)$ up to error $\epsilon = 2^{-\Theta(k(n))}$ using $\poly(n)$ operations and $\bigoh(k(n))$ space.

In Circuit (\ref{circ: poor man phase estimation}) we've assumed $\ket{\psi}$ is an eigenstate of $H$ with eigenvalue $\lambda$. If we measure the control qubit at the end, the probability we obtain 0 is $(1+\cos(\lambda t))/2$. Therefore if $\psi$ is a eigenstate with eigenvalue at most $a$, we can verify this with probability at least $c=(1+\cos(at))/2 - \epsilon$, where $\epsilon$ is the error in the implementation of $\exp(-iHt)$. Otherwise if $\lambda_{min} \ge b$, no state $\psi$ will be accepted with probability more than $s=(1+\cos(bt))/2 + \epsilon$. The separation between $c$ and $s$ is at least 
\begin{align}
(\cos(at)-\cos(bt)) - 2\epsilon &= 2 \sin \left(\frac{(a+b)t}{2}\right) \sin \left(\frac{(b-a)t}{2}\right) - 2\epsilon \ge 2^{-\mathcal{O}(k)}
\end{align}
since $\sin x = \Omega(x)$ for $x \in [0,1]$, $(a+b)t \ge (b-a)t = 2^{-\mathcal{O}(k(n))}$, as long as we choose $\epsilon = 2^{-\Theta(k(n))}$ to be sufficiently small enough. This therefore gives a $\bddQMA{\poly}{\Theta(k(n))}{\Theta(k)}{c}{s}$ protocol for $c - s = 2^{-\mathcal{O}(k(n))}$, as desired.
\end{proof}

\begin{lemma} \label{lem: pspace upper bound}
$\bigcup_{c - s \ge 2^{-\bigoh(k)}}\bddQMA{\poly}{\bigoh(k)}{\bigoh(k)}{c}{s} \subseteq \unitaryBQSPACE{k(n)}.$
\end{lemma}
\begin{proof}[Proof sketch]
We only give a high level overview of the proof here; for the complete proof see Appendix \ref{app: pspace upper bound}. The core of the proof is to develop and use new {\emph{space-efficient}} $\QMA$ error reduction procedures. Our procedures are based on the ``in-place'' $\QMA$ amplification procedure of Marriott and Watrous \cite{mw05}, which allows the error in a $\QMA$ proof system to be reduced without requiring additional copies of the witness state. This was improved by Nagaj, Wocjan, and Zhang \cite{nwz11}, whose phase-estimation based procedure reduces the error to $2^{-r}$ using only $\bigoh\left(r\log{\frac{1}{c-s}}\right)$ additional qubits and $\bigoh(r/(c-s))$ repetitions of the circuit and its inverse. We derive a procedure (Lemma \ref{lem: gap amp}) that gives the same error bounds while using only $\bigoh\left(r+\log{\frac{1}{c-s}}\right)$ additional qubits, but still using only $\bigoh(r/(c-s))$ repetitions of the circuit; the improved space bound will be required for our purposes\footnote{In recent work we improved this result to achieve such amplification using only $\log{\frac{r}{c-s}}$ additional space \cite{fklmn16}.}.

Thus we can amplify the gap in our $\QMA$ protocols to still use $\bigoh(k)$ space, but with completeness $1-2^{-\bigoh(k)}$ and soundness $2^{-\bigoh(k)}$. We can now replace the witness by the completely mixed state (or alternatively half of many EPR pairs), which gives us a computation with \emph{no} witness such that the resulting completeness and soundness are both exponentially small, but are still separated by $2^{-\bigoh(k)}$. Finally, we can once again apply our space-efficient amplification procedure to this witness-free protocol, obtaining a computation in $\unitaryBQSPACE{\bigoh(k)}$.
\end{proof}

\begin{lemma}\label{lem: specham-hardness}
$\spechamiltonian{\bigoh(k(n))}$ is $\unitaryBQSPACE{k(n)}$-hard under classical poly-time $\bigoh(k(n))$-space reductions.
\end{lemma}

\begin{proof}[Proof sketch]
Again we only give an overview; see Appendix \ref{sec: specham-hardness} for the full proof. Recall that our uniformity condition on $\spechamiltonian{k(n)}$ implies that every language in $\spechamiltonian{k(n)}$ can be decided by a quantum circuit of size at most $2^{\bigoh(k(n))}$. We first use our space-efficient error reduction procedure to amplify the gap; then we apply a variant of Kitaev's clock construction \cite{ksv02} to construct a Hamiltonian from this amplified circuit. We use a \emph{binary} clock instead of a unary one to save space; since the number of gates is at most $2^{\bigoh(k(n))}$, the clock only needs to be of size $\bigoh(k(n))$, and the total dimension of the system is $2^{\bigoh(k(n))}$ as required. Therefore the Hamiltonian is not local, but it remains sparse (with only a constant number of nonzero terms in each row). Kitaev's analysis then shows that we can obtain a gap inverse polynomial in the circuit size, or inverse exponential in $k(n)$.
\end{proof}

\begin{proof}[Proof of Theorems \ref{thm: spechamiltonian} and \ref{thm: equivalence}]
Immediate from Lemmas \ref{lem:qma protocol}, \ref{lem: pspace upper bound}, and \ref{lem: specham-hardness}.
\end{proof}

Note the polynomial space case in Theorem \ref{thm: equivalence} is Corollary \ref{cor: preciseqma}, that $\preciseQMA=\PSPACE$.

Finally, we end with two results particular to the polynomial space case. First of all, in the equality $\preciseQMA=\PSPACE$, we can actually achieve perfect completeness ($c=1$) for the $\QMA$ proof protocol, assuming the underlying gate set contains the Hadamard and Toffoli gates. Moreover for perfect completeness we do not require that $c-s > 2^{-\poly}$:
\begin{proposition} \label{prop: perfect completeness} Let $\QMA(c,s) = \bddQMA{\poly}{\poly}{\poly}{c}{s}$. Then
\begin{equation}
\PSPACE = \QMA(1,1-2^{-\poly}) = \bigcup_{s < 1}\QMA(1,s), \nonumber
\end{equation}
where we assume that the gateset we use contains the Hadamard and Toffoli gates. In the last term, the union is taken over all functions $s(n)$ such that $s(n) < 1$ for all $n$.
\end{proposition}
The containment $\QMA(1,s) \subseteq \PSPACE$ is known \cite{ikw12}. We prove this proposition in Appendix~\ref{app:perfectcompleteness}.

Our second result concerns the $\QMA$-complete Local Hamiltonian problem. We show that if we allow the promise gap to be exponentially small, then the problem becomes $\PSPACE$-complete. 
\begin{definition}[\preciseklh]\label{def: precise local hamiltonian}
Given as input is a $k$-local Hamiltonian $H=\sum_{j=1}^rH_j$ acting on $n$ qubits, satisfying $r \in \poly(n)$ and $\|H_j\| \le \poly(n)$, and numbers $a < b$ satisfying $b - a > 2^{-\poly(n)}$. It is promised that the smallest eigenvalue of $H$ is either at most $a$ or at least $b$. Output 1 if the smallest eigenvalue of $H$ is at most $a$, and output 0 otherwise.
\end{definition}
\begin{theorem} \label{thm: precise local hamiltonian}
For any $3 \le k \le \mathcal{O}(\log(n))$, \preciseklh \ is $\preciseQMA$-complete, and hence $\PSPACE$-complete.
\end{theorem}
See Appendix \ref{app:localhamiltonian} for a proof. Combined with the perfect completeness results of Appendix \ref{app:perfectcompleteness}, this will also give a proof that determining whether a local Hamiltonian is frustration-free is a $\PSPACE$-complete problem (Theorem \ref{thm: frustration free} in Appendix \ref{app:localhamiltonian}).

\section{Complete problems for time- and space- bounded classes} \label{sec: time and space}
As we noted in the introduction, variants of the problems we consider are already known to be complete for other time-bounded quantum complexity classes. 
For example, consider the problem of inverting an efficiently encoded $2^{\bigoh(k(n))} \times 2^{\bigoh(k(n))}$ matrix with condition number at most $\kappa(n)$. If $\kappa(n), k(n) = \poly(n)$, this problem is $\BQP$-complete \cite{HHL}. Theorem \ref {thm: matrix invert} says that this problem is instead $\unitaryBQSPACE{\bigoh(k)}$-complete if $\kappa = 2^{\bigoh(k)}$. 
Similarly, consider the problem of determining whether the minimum eigenvalue of an efficiently encoded $2^{\bigoh(k(n))} \times 2^{\bigoh(k(n))}$ matrix is at least $b$ or at most $a$, with $b-a=g(n)$. If $g = 1/\poly$ and $k = \poly$ then this problem is $\QMA$-complete \cite{ksv02,at03}. Theorem \ref{thm: spechamiltonian} says that this problem is instead $\unitaryBQSPACE{\bigoh(k)}$-complete if $g = 2^{-\bigoh(k)}$. 

In both of the problems we consider, we have two parameters that we can vary: for matrix inversion, the condition number $\kappa$ and the matrix size $k$; and for minimum eigenvalue, the promise gap size $g=b-a$ and the matrix size $k$. Varying these two parameters independently gives complete problems for quantum classes that are simultaneously bounded in time and space.

\begin{theorem}
Consider the class of problems solvable by a unitary quantum algorithm using $\poly(T(n))$ gates and $\bigoh(k(n))$ space, where $\Omega(\log(n)) \le k(n) \le T(n) \le 2^{\bigoh(k)} \le 2^{\poly(n)}$. This class has the following complete problem under classical $\poly(n)$-time and $\bigoh(k(n))$-space reductions:

Given as input is the size-$n$ efficient encoding of a $2^{\bigoh(k)} \times 2^{\bigoh(k)}$ positive semidefinite matrix $H$ with a known upper bound $\kappa = \poly(T)$ on the condition number, so that $\kappa^{-1}I\preceq H \preceq I$, and $s,t\in \lbrace 0,1\rbrace^{k(n)}$. It is promised that either $|H^{-1}(s,t)|\geq b$
 or $|H^{-1}(s,t)|\leq a$ for some constants $0 \le a < b \le 1$; determine which is the case.
\end{theorem}

\begin{theorem} \label{thm: time space qma}
For functions $k(n)$, $T(n)$ satisfying $\Omega(\log(n)) \le k(n) \le T(n) \le 2^{\bigoh(k)} \le 2^{\poly(n)}$,
\begin{equation}
\bigcup_{c-s\ge \frac{1}{\poly(T)}} \bddQMA{\poly(n)}{\bigoh(k)}{\bigoh(k)}{c}{s} = \bddQMA{\poly(T)}{\bigoh(k)}{\bigoh(k)}{2/3}{1/3} \nonumber
\end{equation}
Furthermore, the following problem is complete for this class under classical $\poly(n)$-time and $\bigoh(k(n))$-space reductions:

Given as input is the size-$n$ efficient encoding of a $2^{\bigoh(k)} \times 2^{\bigoh(k)}$ PSD matrix $H$, such that 
$\|H\|_{max} = \max_{s,t}|H(s,t)|$ is at most a constant. Let $\lambda_{min}$ be the minimum eigenvalue of $H$. It is promised that either $\lambda_{min} \le a$ or $\lambda_{min} \ge b$, where $a(n)$ and $b(n)$ are numbers such that $b-a \ge 1/\poly(T)$. Output 1 if $\lambda_{min} \le a$, and output 0 otherwise.
\end{theorem}

We omit the proofs; they are straightforward generalizations of the proofs in our paper. These results interpolate between the time-bounded and space-bounded case: when $T = \poly(k)$ the time-bound dominates and we obtain a time-bounded class; while when $T = 2^{\bigoh(k)}$ we obtain a space-bounded class. Note that when $T = 2^{\bigoh(k)}$ then the complexity class in Theorem \ref{thm: time space qma} is equal to $\unitaryBQSPACE{\bigoh(k)}$, as shown in Theorem \ref{thm: equivalence}.

\section{Open Problems}
This work leaves open several questions that may lead to interesting follow-up work:
\begin{compactenum}
\begin{item} Can we use our $\preciseQMA=\PSPACE$ result to prove upper or lower bounds for other complexity classes? \end{item}
\begin{item} Here we have shown $\preciseQMA=\PSPACE$.  Ito, Kobayashi and Watrous  have shown that $\QIP$ with doubly-exponentially small completeness-soundness gap is equal to $\EXP$ \cite{ikw12}.  What can be said about the power of $\QIP$ with exponentially small completeness-soundness gap?\end{item}
\begin{item}In this paper we studied unitary quantum space complexity classes, and showed that \matrixinvert{k(n)} and \spechamiltonian{k(n)} characterize unitary quantum space complexity.  Can similar hardness results be shown for non-unitary quantum space complexity classes?\end{item}
\end{compactenum}

\section{Acknowledgements}
We are grateful to Andrew Childs, Sevag Gharibian, David Gosset, Aram Harrow, Hirotada Kobayashi, Robin Kothari, Tomoyuki Morimae, Harumichi Nishimura, Martin Schwarz, John Watrous, and Xiaodi Wu for helpful conversations, to John Watrous for comments on a preliminary draft, and to anonymous referees for suggestions. This work was supported by the Department of Defense.

\appendix

\section{More details on space-bounded computation}\label{app: space bounded}
For this section, it would be helpful to keep in mind that we always assume the space bound $k(n)$ always satisfies $\Omega(\log(n)) \le k(n) \le \poly(n)$.

We start with the definitions of classical bounded space computation. In discussion of space-bounded classes, we usually consider Turing machines with two tapes, a read-only input tape and a work tape; only the space used on the work tape is counted. For $k:\mathbb{N}\rightarrow\mathbb{N}$, a function $f:\{0,1\}^{*}\rightarrow\{0,1\}^*$ is said to be computable in $k(n)$ space if any bit of $f(x)$ can be computed by a deterministic Turing machine using space $\bigoh(k(|x|))$ on the work tape.  For example, $\Logspace$ is the class of functions that can be computed in $\bigoh(\log{n})$ space.

We now discuss quantum space-bounded complexity classes; for a fuller discussion see \cite{Watrous09}. A straightforward way to define quantum space-bounded classes is to consider a Turing machine with three tapes: a read-only classical input tape, a classical work tape, and a quantum work tape (with two heads) consisting of qubits. This is the model considered in \cite{tashma} and \cite{Watrous03}, except that they allow intermediate measurements (and \cite{Watrous03} allows even more general quantum operations). In this work we consider only computations with no intermediate measurements: we can therefore impose that there are no measurements on the quantum work tape until the register reaches a specified end state, following which a single measurement is performed on the quantum tape and the algorithm accepts or rejects according to the measurement. Therefore the operations performed by the algorithm will not depend on the quantum tape, since there is no way to read information out of it until the end of the algorithm.

Instead of working with Turing machines, in quantum computation it is much more customary (and convenient) to work with quantum circuits. For the setup above, since the operations on the quantum tape are completely classically controlled, we can equivalently consider a quantum circuit generated by a classical space-bounded Turing machine that computes the quantum gates one-by-one and applies them in sequence. If the classical Turing machine is $\mathcal{O}(k(n))$-space bounded, it has at most $2^{\mathcal{O}(k)}$ configurations, and therefore there are at most $2^{\mathcal{O}(k)}$ quantum gates in the circuit. 

Moreover, the $\mathcal{O}(k)$-space bounded classical Turing machine can be replaced by a classical circuit on $\mathcal{O}(k)$ bits, such that there is a $\poly(n)$-time $\mathcal{O}(k)$-space Turing machine that on input $i$ generates the $i$-th gate of the circuit (see e.g. \cite[Section~6.8]{ab09}). The classical circuit can then be bundled into the quantum circuit, and we obtain a quantum circuit with at most $2^{\mathcal{O}(k)}$ gates, such that each individual gate can be generated in classical $\poly(n)$-time and $\mathcal{O}(k)$-space. This justifies the definition of the complexity class $\QSPACE[k(n)](c,s)$:

\begingroup
\def\thedefinition{\ref{def: qspace}}
\begin{definition}
Let $k(n)$ be a function satisfying $\Omega(\log(n)) \le k(n) \le \poly(n)$. A promise problem $L=(L_{yes},L_{no})$ is in $\QSPACE[k(n)](c,s)$ if there exists a $\poly(|x|)$-time $\mathcal{O}(k)$-space uniformly generated family of quantum circuits $\{Q_x\}_{x\in\{0,1\}^*}$, where each circuit $Q_x=U_{x,T}U_{x,T-1}\cdots U_{x,1}$ has $T=2^{\mathcal{O}(k)}$ gates, and acts on $\mathcal{O}(k(|x|))$ qubits, such that:\\
 If $x \in L_{yes}$:
\begin{equation}
\bra{0^k} Q^\dagger_x \ket{1}\bra{1}_{out} Q_x \ket{0^k} \ge c.
\end{equation}
Whereas if $x \in L_{no}$:
\begin{equation}
\bra{0^k} Q^\dagger_x \ket{1}\bra{1}_{out} Q_x \ket{0^k} \le s.
\end{equation}
Here $out$ denotes a single qubit we measure at the end of the computation; no intermediate measurements are allowed.  
Furthermore, we require $c$ and $s$ to be computable in classical $\bigoh(k(n))$-space.
\end{definition}
\addtocounter{theorem}{-1}
\endgroup

\section{Proof that \matrixinvert{\bigoh(k(n))} $\in \unitaryBQSPACE{k(n)}$}\label{app: matrix inversion alg}

\begingroup
\def\thetheorem{\ref{thm: matrix inversion alg}}
\begin{theorem}
Fix functions $k(n)$, $\kappa(n)$, and $\epsilon(n)$. Suppose we are given the size-$n$ efficient encoding of a $2^{k(n)} \times 2^{k(n)}$ PSD matrix $H$ such that $\kappa^{-1} I \preceq H \preceq I$. We are also given $\poly(n)$-time $\mathcal{O}(k+\log(\kappa/\epsilon))$-space uniform quantum circuits $U_a$ and $U_b$ acting on $k$ qubits and using at most $T$ gates. Let $U_a\ket{0}^{\otimes k(n)} = \ket{a}$ and $U_b\ket{0}^{\otimes k(n)} = \ket{b}$. The following tasks can be performed with $\poly(n)$-time $O(k+\log(\kappa/\epsilon))$-space uniformly generated quantum circuits with $\poly(T,k,\kappa,1/\epsilon)$ gates and $\bigoh (k+\log(\kappa/\epsilon))$ qubits:
\begin{compactenum}
\item With at least constant probability, output an approximation of the quantum state $H^{-1}\ket{b} / \|H^{-1}\ket{b}\|$ up to error $\epsilon$.
\item Estimate $\|H^{-1}\ket{b}\|$ up to precision $\epsilon$.
\item Estimate  $|\bra{a}H^{-1}\ket{b}|$ up to precision $\epsilon$.
\end{compactenum}
These circuits do not require intermediate measurements.
\end{theorem}
\addtocounter{theorem}{-1}
\endgroup

For convenience we also restate the following lemma:
\begingroup
\def\thelemma{\ref{lem: matrix inversion lemma}}
\begin{lemma}[Implicit in \cite{HHL,tashma}]
There is a $\poly$-time $\bigoh (k + \log(\kappa/\epsilon'))$-space uniform quantum unitary transformation $W_H$ over $k+\ell = \bigoh (k+ \log(\kappa/\epsilon'))$ qubits and using $\poly(k,\kappa/\epsilon')$ gates, such that for any $k$-qubit input state $\ket{b}$,
\begin{equation}
W_H (\ket{0}^{\otimes \ell} \otimes \ket{b}) = \alpha \ket{0}_{out} \otimes \ket{\psi_b} + \sqrt{1-\alpha^2} \ket{1}_{out} \otimes\ket{\psi'_b},
\end{equation}
where $\ket{\psi_b}$ and $\ket{\psi'_b}$ are normalized states such that $\| \ket{\psi_b} - \ket{0}^{\otimes \ell-1} \otimes \frac{H^{-1}\ket{b}} {\|H^{-1}\ket{b}\|} \| \le \epsilon'$, $\alpha$ is a positive number satisfying $|\alpha - \frac{\|H^{-1}\ket{b}\|}{\kappa}| \le \epsilon'$, and ``out'' is a 1-qubit register.
\end{lemma}
\addtocounter{theorem}{-1}
\endgroup

Before we start the proof, we note that for ease of exposition we will actually use a limited number of intermediate measurements of up to $\bigoh(k+\log(\kappa/\epsilon))$ qubits. These intermediate measurements are not necessary, since they can be deferred to the end of the computation using only $\bigoh(k+\log(\kappa/\epsilon))$ extra space, which fits within our space bound. (It is only when the number of intermediate measurements used is superlinear in the amount of space available that we cannot defer measurements.)

\begin{proof}[Proof of Theorem \ref{lem: matrix inversion lemma}]
We first show the first item, i.e. generating the state $H^{-1}\ket{b} / \|H^{-1}\ket{b}\|$. Choose $\epsilon' = \mathcal{O}(\epsilon/\kappa)$ in the statement of Lemma \ref{lem: matrix inversion lemma}, keeping in mind for the rest of the proof that $\log(\kappa/\epsilon') = \mathcal{O}(\log(\kappa/\epsilon))$. Note that $\ket{\psi_b}$ can be obtained by computing $W_HU_b\ket{0}_{all} = W_H (\ket{0}_{anc} \otimes \ket{b})$, and then postselecting on the output qubit being in state $\ket{0}$. To obtain $\ket{\psi_b}$ with high probability we can repeat this procedure many times until success. We can then get a good approximation to $\frac{H^{-1}\ket{b}} {\|H^{-1}\ket{b}\|} $ by tracing out the other ancilla qubits. For our setting we would like to get by with a low space requirement and without using intermediate measurements, so instead of repeating until success, we will apply amplitude amplification to the above unitary $W_H$. Define the projectors $\Pi_0$ and $\Pi_1$ by
\begin{align}
\Pi_0 &= \ket{0}\bra{0}_{anc} \otimes \ket{b}\bra{b} \\
\Pi_1 &= W_H^\dagger (\ket{0}\bra{0}_{out} \otimes I) W_H
\end{align}
Define $\ket{v} = \ket{0}_{anc}\otimes \ket{b}$, and write
\begin{align}
\ket{v} &= \sin \theta \ket{w} +\cos \theta\ket{w^\perp}, \\ 
\ket{w} &= \sin \theta \ket{v} + \cos \theta\ket{v^\perp}
\end{align}
where $\ket{v^\perp}$, $\ket{w}$, and $\ket{w^\perp}$ are normalized states such that 
\begin{align}
\Pi_0\ket{v} = \ket{v},&\quad \Pi_0\ket{v^\perp} = 0 \\
\Pi_1\ket{w} = \ket{w},&\quad \Pi_1\ket{w^\perp} = 0.
\end{align}
Note that
\begin{equation} 
(\bra{0}_{anc} \otimes I) W_H (\ket{0}\bra{0}_{anc} \otimes \ket{b}\bra{b}) (\ket{0}_{anc} \otimes \ket{b}) = W_H \Pi_1\Pi_0(\ket{0}_{anc} \otimes \ket{b}) \propto W_H \ket{w},
\end{equation}
and therefore $W_H\ket{w}$ is the postmeasurement state we desire. The success probability of the postselection step is simply $\bra{v}\Pi_0\Pi_1\Pi_0\ket{v} = \alpha^2 = \sin^2\theta$. 

Consider now the operator $R=-(I-2\Pi_1)(I-2\Pi_0)$; by analogy from Grover's algorithm, it is easy to see that $R$ is a rotation operator with angle $2\theta$. It has eigenvalues $e^{\pm i2\theta}$, with the following eigenvectors:
\begin{align}
R\ket{\psi_+} = e^{i2\theta}\ket{\psi_+},&\quad \ket{\psi_+} \equiv \frac{1}{\sqrt{2}}\left(\ket{v} + i \ket{v^\perp} \right) = \frac{ie^{-i\theta}}{\sqrt{2}}\left(\ket{w} - i\ket{w^\perp} \right) \label{eq: psi+} \\
R\ket{\psi_-} = e^{-i2\theta}\ket{\psi_-},&\quad \ket{\psi_-} \equiv \frac{1}{\sqrt{2}}\left(\ket{v} - i \ket{v^\perp} \right) = \frac{-ie^{i\theta}}{\sqrt{2}}\left(\ket{w} + i\ket{w^\perp} \right)
\end{align}
Note our initial state $\ket{v}$ is a uniform superposition of the eigenstates of $R$: $\ket{v} = 2^{-1/2} (\ket{\psi_+} + \ket{\psi_-})$. We perform phase estimation of the rotation operator $R$ on the state $\ket{v}$ with precision $\bigoh(\epsilon')$ and failure probability $\delta = 2^{-\bigoh(k)}\poly(\kappa/\epsilon)$; this requires $\poly(T,k,\kappa,1/\epsilon)$ gates and $\bigoh(k+\log(\kappa/\epsilon))$ extra ancilla qubits to perform, and with probability $1-\delta$ outputs an estimate for either $\theta$ or $-\theta$ with error $\bigoh(\epsilon')$. Taking the absolute value of the sine of the output, we obtain an estimate of $\alpha = \sin \theta$ with error $\bigoh(\epsilon') = \bigoh(\epsilon/\kappa)$. Since $|\alpha - \frac{\|H^{-1}\ket{b}\|}{\kappa}| \le \epsilon'$, this allows us to calculate an estimate for $\|H^{-1}\ket{b}\|$ up to precision $\mathcal{O}(\kappa \epsilon') = \mathcal{O}(\epsilon)$, with probability $1-2^{-\bigoh(k)}\poly(\kappa/\epsilon)$. This completes the second task.

To complete the first task, we note that the residual state after the phase estimation procedure is still a linear combination of $\ket{\psi_+}$ and $\ket{\psi_-}$. If the phase estimation procedure above had output an estimate $\tilde{\theta}_+ \approx_{\bigoh(\epsilon')} \theta$ (this happens with probability at least $(1-\delta)/2$), then because the we set the failure probability to be $\delta$, the residual state must be $\delta$-close to $\ket{\psi_+}$.  From \eqref{eq: psi+} we see that this residual state has constant overlap with $\ket{w}$. Similarly, if the phase estimation procedure had instead output an estimate close to $-\theta$ (this happens with probability at least $(1-\delta)/2$), then we obtain a residual state $\delta$-close to $\ket{\psi_-}$, and hence with constant overlap with $\ket{w}$. By applying $W_H$ to our state and verifying the ancilla qubits are all zero, we obtain the desired state $W_H\ket{w}\approx H^{-1}\ket{b}$ with constant probability.

Finally, if an estimate for $|\bra{a} H^{-1} \ket{b}|$ is desired, we can consider instead the following modification to $\Pi_0$ and $\Pi_1$:
\begin{align}
\Pi_0 &= \ket{0}\bra{0}_{anc} \otimes \ket{b}\bra{b} \\
\Pi'_1 &= W_H^\dagger (\ket{0}\bra{0}_{out} \otimes I) (I_{anc} \otimes \ket{a}\bra{a})(\ket{0}\bra{0}_{out} \otimes I) W_H
\end{align}
Since $(\ket{0}\bra{0}_{out} \otimes I) W_H (\ket{0}^{\otimes \ell} \otimes \ket{b}) = \alpha \ket{0}_{out} \otimes \ket{\psi_b}$, we see that 
\begin{equation}
\bra{v}\Pi_0\Pi_1'\Pi_0\ket{v} = \alpha^2\left|\left(\bra{0}^{\otimes \ell-1} \otimes \bra{a}\right) \ket{\psi_b}\right|^2.
\end{equation}
$\bra{v}\Pi_0\Pi_1'\Pi_0\ket{v}$ can be estimated in the same way that $\alpha$ has been estimated. Recalling that $\| \ket{\psi_b} - \ket{0}^{\otimes \ell-1} \otimes \frac{H^{-1}\ket{b}} {\|H^{-1}\ket{b}\|} \| \le \epsilon'$ and $|\alpha - \frac{\|H^{-1}\ket{b}\|}{\kappa}| \le \epsilon'$, this allows us to estimate $|\bra{a} H^{-1} \ket{b}|$ to $\bigoh (\epsilon)$ precision.
\end{proof}

\section{$\matrixinvert{k(n)}$ is hard for $\BQSPACE[k(n)]$}\label{app: matrixinversion-hardness}
We begin our hardness proof by considering the following simple hard problem for $\unitaryBQSPACE{k(n)}$:\begin{definition}[\qca{k(n)}]
	Given as input is the size-$n$ classical algorithm (Turing machine) that generates a quantum circuit $Q$ acting on $k(n)$ qubits with $T = 2^{\mathcal{O}(k(n))}$ 1- or 2-qubit gates: on input $i$, the algorithm outputs the gate $i$-th gate of the circuit in polynomial time and $\bigoh(k(n))$ space. It is promised that either the matrix entry $|\langle {\zero}|Q|{\zero}\rangle| \geq 2/3$ or $|\langle {\zero}|Q|{\zero}\rangle| \leq 1/3$; determine which is the case.
\end{definition}
\begin{lemma}[Implicit in \cite{bbbv,dawsonnielsen}]\label{lem: quantum circuit acceptance}
$\bigoh(\qca{k(n))}$ is $\unitaryBQSPACE{k(n)}$-hard under classical $\poly(n)$-time, $\bigoh(k(n))$-space reductions.
\end{lemma}
\begin{proof}This lemma is implicit in e.g., \cite{bbbv,dawsonnielsen}.  We include the proof here for completeness.  Suppose we are given an $x\in \{0,1\}^n$ and would like to determine if $x\in L_{yes}$ for some $L=\{L_{yes},L_{no}\}\in\unitaryBQSPACE{k}$.  There is a quantum circuit on $k(n)$ qubits, $Q_x=U_TU_{T-1}\cdots U_1$ of size $T=2^{\mathcal{O}(k(n))}$ that decides $x$.  That is, 
\begin{equation}
Q_x|\zero\rangle = \sqrt{p_x} |1\rangle_{out}|\psi_x^1\rangle + \sqrt{1-p_x} |0\rangle_{out} |\psi_x^0\rangle.
\end{equation}
where $out$ indicates the designated output qubit, and $|\psi_x^1\rangle$, $|\psi_x^0\rangle$ are $(k-1)$-qubit states; $p_x$ is the probability that the computation accepts, so $p_x \ge 2/3$ if $x \in L_{yes}$ and $p_x \le 1/3$ if $x \in L_{no}$. Note that the uniformity condition on $\unitaryBQSPACE{k}$ guarantees the existence of a classical algorithm that generates $Q_x$.

We now describe a reduction which creates a related circuit $\tilde{Q}_x$ with a single matrix entry that is proportional to the acceptance probability of $Q_x$.  This new circuit $\tilde{Q}_x$ takes the same number of input qubits as $Q_x$ as well as an additional ancillary qubit.  $\tilde{Q}_x$ runs $Q_x$, then using a single CNOT gate copies the state of the output qubit to the ancillary qubit, flips the ancillary qubit, and finally applies the inverse, $Q_x^{\dagger}$, to the input qubits. It is straightforward to check that
\begin{equation}
\langle 0|\langle \zero|\tilde{Q}_x|\zero\rangle|0\rangle = p_x.
\end{equation}
Therefore knowing the single matrix entry $\langle 0|\langle \zero|\tilde{Q}_x|\zero\rangle|0\rangle$ is sufficient to decide if $x \in L_{yes}$. Moreover, $\tilde{Q}_x$ can be computed from $Q_x$ using polynomial time and $\mathcal{O}(k)$ space, and this completes the proof.
\end{proof}

\begin{theorem}
$\matrixinvert{\bigoh(k(n))}$ is $\unitaryBQSPACE{k(n)}$-hard under classical reductions computable in polynomial time and $\bigoh (k(n))$ space.
\end{theorem}
\begin{proof}
We will show that $\matrixinvert{k(n)}$ is as hard as $\qca{k(n)}$.  Given an instance of the latter, i.e., a circuit on $k(n)$ qubits, $Q=U_TU_{T-1}\cdots U_1$ with $T = 2^{\mathcal{O}(k(n))}$, define the following unitary of dimension $3T2^k$:
\[U=\sum_{t=1}^{T}|t+1\rangle\langle t|\otimes U_t+|t+T+1\rangle\langle t+T|\otimes I+|t+2T+1\bmod{3T}\rangle\langle t+2T|\otimes U_{T+1-t}^{\dagger}\]	
Crucially, note that for any $t$ in the range $[T,2T]$ and any state $|\psi\rangle$ on $k(n)$ qubits: 
\begin{equation}\label{eqn:fullclock}
	U^t|1\rangle|\psi\rangle=|t+1\rangle\otimes Q|\psi\rangle
\end{equation}
Furthermore $U^{3T}=I$, a fact we will soon exploit. We now construct the Hermitian matrix:
\begin{equation}
H=\begin{bmatrix}
    0 & I-Ue^{-1/T}\\
   I-U^{\dagger}e^{-1/T} &  0
\end{bmatrix}
\end{equation}
 $H$ has dimension $N=6T2^k$ and condition number $\kappa\leq 2T=2^{\bigoh (k)}$.  Notice that given as input a description of $Q$ we can compute each entry of $H$ to within $2^{-\mathcal{O}(k(n))}$ error in $\mathcal{O}(k(n))$ space, via space efficient algorithms for exponentiation \cite{reif}.
   Furthermore, $H^{-1}$ is the following matrix:
\begin{equation}
\begin{bmatrix}
    0 & \left(I-U^\dagger e^{-1/T}\right)^{-1}\\
   \left(I-Ue^{-1/T}\right)^{-1} &  0
\end{bmatrix}
\end{equation}
$\left(I-Ue^{-1/T}\right)^{-1}$ is just the power series $\sum_{j=0}^\infty U^j e^{-j/T}=\frac{1}{1-e^{-3}}\sum_{j=0}^{3T-1} U^{j} e^{-j/T}$, where we've used $U^{3T}=I$. Therefore for any fixed $t \in [0,3T-1]$,
 \begin{align}
\langle\zero|\langle t+1|\left(I-Ue^{-1/T}\right)^{-1}|1\rangle|\zero\rangle
&= \frac{1}{1-e^{-3}}\langle\zero|\langle t+1|\left(\sum_{j'=0}^{3T-1} U^{j'} e^{-j'/T}\right)|1\rangle|\zero\rangle
\end{align}
which is a particular entry in $H^{-1}$. In the second line we've used $j = 3Tx + j'$ for some integers $x$ and $j' \in [0,3T-1]$, and that $U^{3T}=I$.  For any $t \in [T,2T]$, as a consequence of Equation \ref{eqn:fullclock} the above quantity equals
\begin{equation}
\frac{e^{-t/T}}{1-e^{-3}}\langle\zero|Q|\zero\rangle.
\end{equation}
In particular, an estimation of this entry of $H^{-1}$ will solve \qca{k(n)}.
\end{proof}

\section{Proof of Lemma \ref{lem: pspace upper bound}} \label{app: pspace upper bound}
\subsection{In-place gap amplification of $\QMA$ protocols with phase estimation}\label{app: space efficient amplification}

We start out by proving the following lemma, which proves ``in-place'' gap amplification of $\QMA$ using phase estimation (see also the similar result of Nagaj et. al, Lemma \ref{lem: gap amp 1} in Appendix \ref{app: gap amplification}).
\begin{lemma} \label{lem: gap amp 2}
For any functions $t,k,r>0$, 
\[
\bddQMA{t}{k}{m}{c}{s}\subseteq\bddQMA{\mathcal{O}\left(\frac{t2^r}{c-s}\right)}{\mathcal{O}\left(k+r+\log\left(\frac{1}{c-s}\right)\right)}{m}{1-2^{-r}}{2^{-r}}.
\]
\end{lemma}
\begin{proof}
	Let $L=(L_{yes}, L_{no})$ be a promise problem in $\QMA(c,s)$ and $\{V_x\}_{x\in\{0,1\}^*}$ the corresponding uniform family of verification circuits.
Define the projectors:
\begin{equation}
\Pi_0 = I_m \otimes \ket{0^k}\bra{0^k}, \quad \Pi_1 = V^\dagger_x \left(\ket{1}\bra{1}_{out} \otimes I_{m+k-1}\right) V_x
\end{equation}
and the corresponding reflections
$R_0 = 2\Pi_0 - I, R_1 = 2\Pi_1 - I$.
Define $\phi_c = \arccos\sqrt{c}/\pi$ and $\phi_s = \arccos\sqrt{s}/\pi$ (recalling that these functions can be computed to precision $\bigoh (c-s)$ in space $\bigoh (\log[1/(c-s)])$. 
Now consider the following procedure:
\begin{compactenum}
\item Perform phase estimation of the operator $R_1R_0$ on the state $\ket{\psi}\otimes \ket{0^k}$, with precision $\bigoh (c-s)$ and failure probability $2^{-r}$.
\item Output YES if the phase is at most $(\phi_{c}+\phi_{s})/2$; otherwise output NO.
\end{compactenum}
Phase estimation of an operator $U$ up to precision $a$ and failure probability $\epsilon$ requires $\alpha := \lceil\log_2(1/a)\rceil + \log_2[2+1/(2\epsilon)]$ additional ancilla qubits and $2^\alpha = \mathcal{O}(1/(a\epsilon))$ applications of the control-$U$ operation (see e.g. \cite{nc00}).  Thus, the above procedure can be implemented by a circuit of size $\mathcal{O}(2^{r}t/(c-s))$ using $\mathcal{O}(r+\log[1/(c-s)])$ extra ancilla qubits. Using the standard analysis of in-place $\QMA$ error reduction \cite{mw05,nwz11}, it can be shown that this procedure has completeness probability at least $1-2^{-r}$ and soundness at most $2^{-r}$.
\end{proof}

In Appendix \ref{app: gap amplification} we will prove the following stronger error reduction lemma that gives the same space bound but uses less time. This better time bound will be required for proving Lemma \ref{lem: specham-hardness}.

\begin{lemma} \label{lem: gap amp}
For any functions $t,k,r>0$, 
\[
\bddQMA{t}{k}{m}{c}{s}\subseteq\bddQMA{\mathcal{O}\left(\frac{rt}{c-s}\right)}{\mathcal{O}\left(k+r+\log\left(\frac{1}{c-s}\right)\right)}{m}{1-2^{-r}}{2^{-r}}.
\]
\end{lemma}

Thus, we get the following corollaries:
\begin{corollary} \label{obvious2} For any $r = \mathcal{O}(k)$,
$
\unitaryQSPACE{k}{c}{c-2^{-\bigoh (k)}} \subseteq
\unitaryQSPACE{\Theta(k)}{1-2^{-r}}{2^{-r}}$.
\end{corollary}
This corollary shows that error reduction is possible for unitary quantum $\bigoh(k)$-space bounded classes, as long as the completeness-soundness gap is at least $2^{-\bigoh (k)}$.
\begin{proof}
This follows from Lemma \ref{lem: gap amp} by taking $m=0$, $s = c-2^{-\Theta(k)}$, and $r = \Theta(k)$.
\end{proof}
\begin{corollary}\label{obvious1}
\[
\bddQMA{t}{k}{m}{c}{c-2^{-\Theta(k)}}\subseteq\bddQMA{\mathcal{O}\left(t2^{\Theta(k)}\right)}{\mathcal{O}\left(k\right)}{m}{1-2^{-(m+2)}}{2^{-(m+2)}}.
\]
\end{corollary}
\begin{proof}
This follows from Lemma \ref{lem: gap amp} by taking $s = c-2^{-\Theta(k)}$ and $r = m+2$.
\end{proof}
\subsection{Removing the witness of an amplified $\QMA$ protocol}\label{app: removingwitness}
\begin{theorem} \label{thm:pqpspace simulation} For any function $t=2^{\bigoh(k+m)}$,
\[
\bddQMA{t}{k}{m}{1-2^{-(m+2)}}{2^{-(m+2)}}\subseteq
\unitaryQSPACE{k+m}{3/4\cdot 2^{-m}}{1/4 \cdot 2^{-m}}.
\]
\end{theorem}
\begin{proof}
The proof is very similar to that of \cite[Theorem 3.6]{mw05}. For any functions $m, k$, consider a problem $L\in\bddQMA{t}{k}{m}{1-2^{-(m+2)}}{2^{-(m+2)}}$, and let $\{V'_x\}_{x\in\{0,1\}^*}$ be a uniform family of verification circuits for $L$ with completeness $1-2^{-(m+2)}$ and soundness $2^{-(m+2)}$. 

For convenience, define the $2^m \times 2^m$ matrix:
\begin{equation}
Q_x := \left(I_{2^m}\otimes \bra{0^p}\right) V'^\dagger_x \ket{1}\bra{1}_{out} V'_x \left(I_{2^m}\otimes \ket{0^p}\right).
\end{equation}
$Q_x$ is positive semidefinite, and $\bra{\psi}Q_x\ket{\psi}$ is the acceptance probability of $V'_x$ on witness $\psi$. Thus
\begin{equation}
x\in L_{yes} \Rightarrow \tr[Q_x]\ge 1 - 2^{-(m+2)} \ge 3/4
\end{equation}
since the trace is at least the largest eigenvalue, and $m\geq 0$; likewise,
\begin{equation}
x\in L_{no} \Rightarrow \tr[Q_x]\le 2^m \cdot 2^{-(m+2)} = 1/4
\end{equation}
since the trace is the sum of the $2^m$ eigenvalues, each of which is at most $2^{-(m+2)}$. 

Therefore our problem reduces to determining whether the trace of $Q_x$ is at least $3/4$ or at most $1/4$.  Now we show that using the totally mixed state $2^{-m}I_m$ (alternatively, preparing $m$ EPR pairs and taking a qubit from each pair) as the witness of the verification procedure encoded by $Q_x$, succeeds with the desired completeness and soundness bounds.  The acceptance probability is given by
$\tr(Q_x 2^{-m}I_m) = 2^{-m} \tr(Q_x)$,
which is at least $2^{-m} \cdot 3/4$ if $x\in L_{yes}$, and at most $2^{-m} \cdot 1/4$ if $x\in L_{no}$. Thus we have reduced the problem $L$ to determining if a quantum computation with \emph{no} witness, acting on $k+m$ qubits, accepts with probability at least $3/4 \cdot 2^{-m}$ or at most $s'=1/4 \cdot 2^{-m}$.
\end{proof}
We can finally finish the proof of Lemma \ref{lem: pspace upper bound}.
\begin{proof}[Proof of Lemma \ref{lem: pspace upper bound}]
This follows from Corollary \ref{obvious1}, Theorem \ref{thm:pqpspace simulation}, and Corollary \ref{obvious2}.
\end{proof}

\section{$\spechamiltonian{\bigoh(k(n))}$ is hard for $\unitaryBQSPACE{k(n)}$} \label{sec: specham-hardness}
In this section we prove Lemma \ref{lem: specham-hardness}, stated here for convenience.
\begingroup
\def\thelemma{\ref{lem: specham-hardness}}
\begin{lemma}
$\spechamiltonian{\bigoh(k(n))}$ is $\unitaryBQSPACE{k(n)}$-hard under classical poly-time $\bigoh(k(n))$-space reductions.
\end{lemma}
\addtocounter{theorem}{-1}
\endgroup

\begin{proof}
Let $L=(L_{yes},L_{no})$ be a problem in $\unitaryBQSPACE{k(n)}$, and suppose it has a verifier that uses $t=2^{\bigoh(k)}$ gates with completeness $2/3$ and soundness $1/3$. By Lemma \ref{lem: gap amp}, we can   amplify the gap to get a new verifier circuit that uses $T=\mathcal{O}(rt)$ gates, and has completeness $c=1-2^{-r}$ and soundness $s=2^{-r}$. Choose $r=\bigoh(k)$ large enough so that $c \ge 1 - 1/T^3$ and $s \le 1/T^3$, and suppose $V_x=V_{x,T}V_{x,T-1}\cdots V_{x,1}$ is the new (gap-amplified) verifier circuit for $L$ acting on $k$ qubits. Consider the Kitaev clock Hamiltonian:
\begin{equation}
H = H_{in} + H_{prop} + H_{out}
\end{equation}
defined on the Hilbert space $\mathbb{C}^{2^k} \otimes \mathbb{C}^{T+1}$, where
\begin{equation}
H_{in} = \ket{0^k}\bra{0^k} \otimes \ket{0}\bra{0}, \quad
H_{out} = (\ket{1}\bra{1}_{out} \otimes I_{k-1}) \otimes \ket{T}\bra{T}
\end{equation}
\begin{equation}
H_{prop} = \sum_{j=1}^T \frac{1}{2} \left[ - V_{x,j} \otimes \ket{j}\bra{j-1} - V_{x,j}^\dagger \otimes \ket{j-1}\bra{j} + I \otimes (\ket{j}\bra{j} + \ket{j-1}\bra{j-1}) \right].
\end{equation}
$H$ is a sparse matrix - in fact, there are only a constant number of nonzero terms in each row. Since each gate $V_{x,j}$ can be computed in classical polynomial time and $k(n)$ space, it follows that so can the nonzero entries of $H$. Moreover, let $\lambda_{min}$ be the minimum eigenvalue of $H$; then it was shown by Kitaev \cite{ksv02} that if $x \in L_{yes}$ then $\lambda_{min} \le a = (1-c)/(T+1)$, while if $x \in L_{no}$ then $\lambda_{min} \ge b = (1-s)/T^3$. Since $c \ge 1-1/T^3$ and $s \le 1/T^3$ we see that $2^{-\mathcal{O}(k)} < b-a$. 
\end{proof}

\section{Achieving Perfect Completeness for $\preciseQMA$}\label{app:perfectcompleteness}
We now consider the problem of achieving perfect completeness for $\preciseQMA$. Specifically, we will show the following:
\begingroup
\def\theproposition{\ref{prop: perfect completeness}}
\begin{proposition} Let $\QMA(c,s) = \bddQMA{\poly}{\poly}{\poly}{c}{s}$. Then
\[
\PSPACE = \QMA(1,1-2^{-\poly}) = \bigcup_{s < 1}\QMA(1,s),
\]
where we assume that the gateset we use contains the Hadamard and Toffoli gates. In the last term, the union is taken over all functions $s(n)$ such that $s(n) < 1$ for all $n$.
\end{proposition}
\addtocounter{theorem}{-1}
\endgroup

Since $\PSPACE = \preciseQMA$, this proposition shows that any $\preciseQMA$ protocol can be reduced to a different $\preciseQMA$ protocol with perfect completeness, i.e. in the YES case Arthur accepts Merlin's witness with probability 1. The reduction is rather roundabout, however, and it would be interesting to see if a more direct reduction can be found.

The second equality follows from the first equality and the result by \cite{ikw12} that $\QMA(1,s) \subseteq \PSPACE$. We will therefore only prove the first equality.

Looking back at Circuit \ref{circ: poor man phase estimation}, we see that we \emph{almost} have perfect completeness in our protocol already - if the Hamiltonian simulation of $e^{-iHt}$ could be done without error, then indeed the protocol has perfect completeness. Our strategy will be perform a different unitary that can be performed exactly, but, like $e^{-iHt}$, also allows us to use phase estimation to distinguish the eigenvalues of $H$.

Given a sparse Hamiltonian $H$ (with at most $d$ nonzero entries per row) and a number $X \ge \max_{j,\ell}|H_{j\ell}|$ that upper bounds the absolute value of entries of $H$, Andrew Childs defined an efficiently implementable quantum walk \cite{berry14,childs10}. Each step of the quantum walk is a unitary $U$ with eigenvalues $e^{i\tilde{\lambda}}$, where 
\begin{equation}
\tilde{\lambda} = \arcsin \frac{\lambda}{Xd}
\end{equation}
with $\lambda$ representing eigenvalues of $H$. Note that the YES case $\lambda = 0$ corresponds to $\tilde{\lambda}=0$, and the NO case $\lambda \ge 2^{-g(n)}$ corresponds to $\tilde{\lambda} \ge 2^{-g(n)}/(Xd)$ since $\arcsin x \ge x$ for $|x| \le 1$. In the latter case the $\tilde{\lambda}$ can be at most exponentially small, and therefore the stripped down version of phase estimation still suffices to tell the two cases apart with exponentially small probability.

We now note that the Hamiltonian $H$ we obtain from the hardness reduction from $\PSPACE$ (Lemma \ref{lem: specham-hardness}) is of a very special form. Specifically, since $\unitaryBQSPACE{\poly}=\PSPACE$, we can assume the verifier circuit $V_x$ is deterministic, so it has completeness 1 and soundness 0. Moreover, all of its gates are classical, and hence all entries of the Kitaev clock Hamiltonian $H$ are $0$, $\pm 1/2$, or $1$.

For the matrix $H$ satisfying the above, $U$ can be implemented exactly with a standard gateset; perfect completeness of the protocol will then follow. If $H$ is a $N \times N$ matrix (where $N=2^n$), $U$ is (see presentation in \cite[Section~3.1~and~Lemma~10]{berry15}) a unitary defined on the enlarged Hilbert space $\mathbb{C}^{2N} \otimes \mathbb{C}^{2N} = (\mathbb{C}^{N} \otimes \mathbb{C}^2) \otimes  (\mathbb{C}^{N} \otimes \mathbb{C}^2)$, as follows:
\begin{equation}
U = ST(I_{2N} \otimes (I_{2N} - 2\ket{0}\bra{0}_{2N})) T^\dagger
\end{equation}
where the $2N$ subscript indicates a register of dimension $2N$, the unitary $S$ swaps the two registers, and the unitary $T$ is defined by
\begin{equation}
T = \sum_{j=0}^{N-1}\sum_{b \in \{0,1\}} (\ket{j}\bra{j} \otimes \ket{b}\bra{b}) \otimes \ket{\varphi_{jb}}\bra{0}_{2N}
\end{equation}
with $\ket{\varphi_{j1}} = \ket{0}_N\ket{1}$ and
\begin{equation}
\ket{\varphi_{jb}} = \frac{1}{\sqrt{d}} \sum_{\ell \in F_j} \ket{\ell}\left(\sqrt{\frac{H^*_{j\ell}}{X}}\ket{0} + \sqrt{1-\frac{|H^*_{j\ell}|}{X}}\ket{1}\right),
\end{equation}
where $F_j$ index the nonzero entries in the $j$-th row. Recall that for any $j,\ell$, $H_{j\ell}=0$, $\pm 1/2$, or $1$, and hence we can take $X=1$. If we furthermore assume $d$ is a power of 2 (which we can always do by adding indices of zero entries to $F_j$), it is straightforward to see that both $S$ and $T$ can be implemented using just Hadamard gates and classical gates (Pauli-$X$, controlled-$X$, and Toffoli gates) - the latter of which can be implemented using just Toffoli gates and access to a qubit in the $\ket{1}$ state (which can be provided by the prover). Therefore $U$ can be exactly implemented in any gateset that allows Hadamard gates and Toffoli gates to be implemented exactly.

\section{Precise Local Hamiltonian Problem}\label{app:localhamiltonian}
Recall the following definition:
\begingroup
\def\thedefinition{\ref{def: precise local hamiltonian}}
\begin{definition}[\preciseklh]
Given as input is a $k$-local Hamiltonian $H=\sum_{j=1}^rH_j$ acting on $n$ qubits, satisfying $r \in \poly(n)$ and $\|H_j\| \le \poly(n)$, and numbers $a < b$ satisfying $b - a > 2^{-\poly(n)}$. It is promised that the smallest eigenvalue of $H$ is either at most $a$ or at least $b$. Output 1 if the smallest eigenvalue of $H$ is at most $a$, and output 0 otherwise.
\end{definition}
\endgroup
In this section we will prove the following:
\begingroup
\def\thetheorem{\ref{thm: precise local hamiltonian}}
\begin{theorem}
For any $3 \le k \le \mathcal{O}(\log(n))$, \preciseklh \ is $\preciseQMA$-complete, and hence $\PSPACE$-complete.
\end{theorem}
\endgroup
\begin{proof}
This proof follows straightforwardly by adapting the proof of \cite{ksv02} and \cite{kr03}. The proof of containment in $\preciseQMA$ is identical to the containment of the usual Local Hamiltonian problem in $\QMA$; see \cite{ksv02} for details.

To show $\preciseQMA$-hardness, we note that for a $\QMA$-verification procedure with $T$ gates, completeness $c$ and soundness $s$, \cite{kr03} reduces this to a 3-local Hamiltonian with lowest eigenvalue no more than $(1-c) / (T+1)$ in the YES case, or no less than $(1-s) / T^3$ in the NO case. For this to specify a valid \preciselh \ problem we need that
\begin{equation} \label{eq:preciselh_condition}
\frac{1-s}{T^3} - \frac{1-c}{T+1} > 2^{-\poly(n)}.
\end{equation}
Recalling that we showed that perfect completeness can be assumed for $\preciseQMA$-hard problems, we can take $c=1$, $s = 1-2^{-\poly(n)}$ and the above inequality trivially holds. Hence any problem in $\PSPACE$ can be reduced to a \preciseilh{3} problem.
\end{proof}
In fact, even just testing whether a $k$-Local Hamiltonian is frustration-free is $\PSPACE$-complete:\begin{definition} [\emph{Frustration-Free $k$-Local Hamiltonian}] Given as input is a $k$-local Hamiltonian $H=\sum_{j=1}^rH_j$ acting on $n$ qubits, satisfying $r \in \poly(n)$, each term $H_j$ is positive semidefinite, and $\|H_j\| \le \poly(n)$. Output 1 if the smallest eigenvalue of $H$ is zero, and output 0 otherwise.
\end{definition}
\begin{theorem} \label{thm: frustration free}
Frustration-Free $k$-Local Hamiltonian is $\PSPACE$-complete.
\end{theorem}
\begin{proof}
The containment in $\PSPACE$ follows from the proof of the containment of the usual Local Hamiltonian problem in $\QMA$ \cite{ksv02}, along with Proposition \ref{prop: perfect completeness}. $\PSPACE$-hardness follows from the proof of Theorem \ref{thm: precise local hamiltonian}, by taking $c=1$ in the proof.
\end{proof}
\section{In-place gap amplification} \label{app: gap amplification}
In this appendix we will prove Lemma \ref{lem: gap amp}. To do so we first start out with the following weaker result:
\begin{lemma}[Implicit in Nagaj, Wocjan, and Zhang \cite{nwz11}] \label{lem: gap amp 1}
For any functions $t,k,r>0$, 
\[
\bddQMA{t}{k}{m}{c}{s}\subseteq\bddQMA{\mathcal{O}\left(\frac{rt}{c-s}\right)}{\mathcal{O}\left(k+r\log\left(\frac{1}{c-s}\right)\right)}{m}{1-2^{-r}}{2^{-r}}.
\]
\end{lemma}
\begin{proof}
	Let $L=(L_{yes}, L_{no})$ be a promise problem in $\QMA(c,s)$ and $\{V_x\}_{x\in\{0,1\}^*}$ the corresponding uniform family of verification circuits.
Define the projectors:
\begin{align}
\Pi_0 &= I_m \otimes \ket{0^k}\bra{0^k} \\
\Pi_1 &= V^\dagger_x \left(\ket{1}\bra{1}_{out} \otimes I_{m+k-1}\right) V_x
\end{align}
and the corresponding reflections:
\begin{equation}
R_0 = 2\Pi_0 - I, \quad R_1 = 2\Pi_1 - I.
\end{equation}
Define $\phi_c = \arccos\sqrt{c}/\pi$ and $\phi_s = \arccos\sqrt{s}/\pi$ (recalling that these functions can be computed to precision $\bigoh (c-s)$ in space $\bigoh (\log[1/(c-s)])$. 
Now consider the following procedure:
\begin{compactenum}
\item Perform $r$ trials of phase estimation of the operator $R_1R_0$ on the state $\ket{\psi}\otimes \ket{0^k}$, with  precision $\mathcal{O}(c-s)$ and $1/16$ failure probability. 
\item If the median of the $r$ results is at most $(\phi_{c}+\phi_{s})/2$, output YES; otherwise output NO.
\end{compactenum}
Phase estimation of an operator $U$ up to precision $a$ and failure probability $\epsilon$ requires $\alpha := \lceil\log_2(1/a)\rceil + \log_2[2+1/(2\epsilon)]$ additional ancilla qubits and $2^\alpha = \mathcal{O}(1/(a\epsilon))$ applications of the control-$U$ operation (see e.g. \cite{nc00}).  Thus, the above procedure, which uses $r$ applications of phase estimation to precision $\mathcal{O}(c-s)$, can be implemented by a circuit of size $\mathcal{O}(rt/(c-s))$ using $\mathcal{O}(r\log[1/(c-s)])$ extra ancilla qubits. Using the standard analysis of in-place $\QMA$ error reduction \cite{mw05,nwz11}, it can be seen that this procedure has completeness probability at least $1-2^{-r}$ and soundness at most $2^{-r}$.
\end{proof}
We can now prove Lemma \ref{lem: gap amp}, which we restate below:
\begingroup
\def\thelemma{\ref{lem: gap amp}}
\begin{lemma}
For any functions $t,k,r>0$, 
\[
\bddQMA{t}{k}{m}{c}{s}\subseteq\bddQMA{\mathcal{O}\left(\frac{rt}{c-s}\right)}{\mathcal{O}\left(k+r+\log\left(\frac{1}{c-s}\right)\right)}{m}{1-2^{-r}}{2^{-r}}.
\]
\end{lemma}
\addtocounter{theorem}{-1}
\endgroup

\begin{proof}
\begin{align}
\bddQMA{t}{k}{m}{c}{s} &\subseteq \bddQMA{\mathcal{O}\left(\frac{t}{c-s}\right)}{\mathcal{O}\left(k+\log\left(\frac{1}{c-s}\right)\right)}{m}{3/4}{1/4}  \nonumber \\
&\subseteq \bddQMA{\mathcal{O}\left(\frac{rt}{c-s}\right)}{\mathcal{O}\left(k+r+\log\left(\frac{1}{c-s}\right)\right)}{m}{1-2^{-r}}{2^{-r}} \nonumber
\end{align}
where the first line follows by taking $r=2$ in Lemma \ref{lem: gap amp 2}, and the second line follows from Lemma~\ref{lem: gap amp 1}.
\end{proof}

\section{Proof sketch of $\PQP^{O_{PEPS}}_{\parallel,\text{classical}} = \PP$} \label{app:peps}
Since $\PP \subseteq \BQP^{O_{PEPS}}_{\parallel,\text{classical}} \subseteq \PQP^{O_{PEPS}}_{\parallel,\text{classical}}$ \cite{swv07}, we only need to show that $\PQP^{O_{PEPS}}_{\parallel,\text{classical}} \subseteq \PP$. In \cite{swv07} it was noted that all PEPS can be seen as the output of a quantum circuit followed by a postselected measurement. Therefore $\PQP^{O_{PEPS}}_{\parallel,\text{classical}}$ corresponds to the problems that can be decided by a quantum circuit, followed by a postselected measurement (since the queries to $O_{PEPS}$ are classical and nonadaptive, we can compose them into one single postselection), followed by a measurement. In the YES case the measurement outputs 1 with probability at least $c$, whereas in the NO case the measurement outputs 1 with probability at most $s$, with $c > s$. The standard counting argument placing $\BQP$ inside $\PP$ then applies to this case as well; see for instance \cite[Propositions~2~and~3]{aaronson05}.

 \bibliography{completespace}
\bibliographystyle{plain}

\end{document}